\documentclass{article}

\usepackage{arxiv}

\usepackage{amsthm,amsmath,amssymb,amsfonts,amscd}
\usepackage[utf8]{inputenc} 
\usepackage[T1]{fontenc}    
\usepackage{hyperref}       
\usepackage{url}            
\usepackage{booktabs}       
\usepackage{amsfonts}       
\usepackage{nicefrac}       
\usepackage{microtype}      
\usepackage{lipsum}		
\usepackage{graphicx}
\usepackage{doi}

\theoremstyle{thmstyleone}%
\newtheorem{theorem}{Theorem}%
\newtheorem{property}{Property}%

\theoremstyle{thmstylethree}%
\newtheorem{definition}{Definition}%

\title{Voronoi Diagrams of Arbitrary Order on the Sphere}

\author{Mercè Claverol \\
	Departament de Matemàtiques\\
	Universitat Politècnica de Catalunya\\
	Barcelona, Spain \\
	\texttt{merce.claverol@upc.edus} \\
\And Andrea de las Heras Parrilla \\
	Departament de Matemàtiques\\
	Universitat Politècnica de Catalunya\\
	Barcelona, Spain \\
	\texttt{andrea.de.las.heras@estudiantat.upc.edu} \\
\And Clemens Huemer \\
	Departament de Matemàtiques\\
	Universitat Politècnica de Catalunya\\
	Barcelona, Spain \\
	\texttt{clemens.huemer@upc.edu} \\}

\renewcommand{\headeright}{}
\renewcommand{\undertitle}{}
\renewcommand{\shorttitle}{Voronoi Diagrams of Arbitrary Order on the Sphere}

\hypersetup{
pdftitle={A template for the arxiv style},
pdfsubject={q-bio.NC, q-bio.QM},
pdfauthor={David S.~Hippocampus, Elias D.~Striatum},
pdfkeywords={First keyword, Second keyword, More},
}

\begin{document}
\maketitle

\begin{abstract}
	For a given set of points $U$ on a sphere $S$, the order $k$ spherical Voronoi diagram $SV_k(U)$ decomposes the surface of $S$ into regions whose points have the same $k$ nearest points of $U$. 
    Hyeon-Suk Na, Chung-Nim Lee, and Otfried Cheong (Comput. Geom., 2002) applied inversions to construct $SV_1(U)$. We generalize their construction for spherical Voronoi diagrams from order $1$ to any order $k$. We use that construction to prove formulas for the numbers of vertices, edges, and faces in $SV_k(U)$. These formulas were not known before.
    We obtain several more properties for $SV_k(U)$, and we also show that $SV_k(U)$ has a small orientable cycle double cover.
\end{abstract}


\section{Introduction}
    
    Let $U$ be a set of $n$ points on a sphere $S \subset \mathbb{R}^3$ such that no three of them lie in the same great circumference and no four of them are cocircular, i.e. $U$ is in general position, and let $1\leq k \leq n-1$ be an integer. 
    The order $k$ spherical Voronoi diagram $SV_k(U)$ decomposes the surface of $S$ into regions whose points have the same $k$ nearest points of $U$. Then, each of these regions is a face $f(P_k)$ of $SV_k(U)$ associated with a subset $P_k\subset U$ of size $k$: Each point in the interior of $f(P_k)$ has $P_k$ as its $k$ nearest neighbors from $U$. Vertices of $SV_k(U)$ are of type I if they are centers of circles on the sphere passing through three points of $U$ and enclosing $k-1$ points of $U$, and are of type II, if they are centers of circles on the sphere passing through three points of $U$ and enclosing $k-2$ points of $U$. Equivalently, the circle defining a vertex of type I (type II) passes through one (two) points of $P_k$ and through two (one) points of $U\backslash P_k$. In the literature, see e.g.~\cite{L82}, vertices of type I (type II) are often called {\it{new}} and {\it{old}}. 
    
    Many researchers studied the nearest ($k=1$) and the farthest ($k=n-1$) spherical Voronoi diagrams, see e.g.~\cite{B79,NLC02, M1971}. Specific algorithms for constructing Voronoi diagrams of order one on the sphere were given in~\cite{augenbaum1985construction, caroli2010robust, chen2003algorithm, dinis2010sweeping, renka1997algorithm, zheng2011plane}. Spherical Voronoi diagrams can also be obtained using the algorithms for Voronoi diagrams in the plane, via inversions~\cite{B78, B79, NLC02}. On the other hand, they can also be constructed by mapping the points to higher dimension and using convex hull algorithms~\cite{aurenhammer1990new}, or hyperplane arrangements~\cite{E87}. We also refer to~\cite{LPL15} for a discussion on several algorithms for constructing Voronoi diagrams.

		
		Spherical Voronoi diagrams of order different from $k=1$ and $k=n-1$ have barely been studied. In this work we deepen in these diagrams and the properties and algorithm that we present are for Voronoi diagrams of arbitrary order $k$ on the sphere. 
    One of the most important tools that we use in our proofs is an edge labeling for $SV_k(U)$. This labeling is an extension to the sphere of the edge labeling for Voronoi diagrams in the plane~\cite{MACA21}.
    An edge that delimits a face of $SV_k(U)$ is a spherical segment of the perpendicular bisector (on the sphere) of two points $i$ and $j$ of $U$. 
    This observation induces a natural labeling of the edges of $SV_k(U)$ with the following rule:
    
    $\bullet$ {\bf{Edge rule:}}\label{edgerule}
    An edge of $SV_k(U)$ which belongs to the perpendicular bisector of points $i,j \in U$ has labels $i$ and $j$, where label $i$ is on the side (half-sphere) of the edge that contains point $i$, and label $j$ is on the other side. See Figure~\ref{fig:F-bisectors}.\\ 
		
From this rule, we deduce two more rules of the labeling of $SV_k(U)$: one rule for the vertices and one rule for the faces, whose proofs are straight-forward and also given in~\cite{MACA21}. 
    
    $\bullet$ {\bf{Vertex rule:}}
    Let $v$ be a vertex of $SV_k(U)$ and let $\{i, j, \ell\}\subset U$ be the set of labels of the edges incident to $v$. 
    The cyclic order of the labels of the edges around $v$ is $i, i, j, j, \ell, \ell$ if $v$ is of type I, and it is $i, j, \ell, i, j, \ell$ if $v$ is of type II.
      
    $\bullet$ {\bf{Face rule:}}
    In each face of $SV_k(U)$, the edges that have the same label $i$ are consecutive, and these labels $i$ are either all in the interior of the face, or are all in the exterior of the face.
    
    \begin{figure}[h!]	
    	\begin{center}
    		\includegraphics[width=0.4\textwidth]{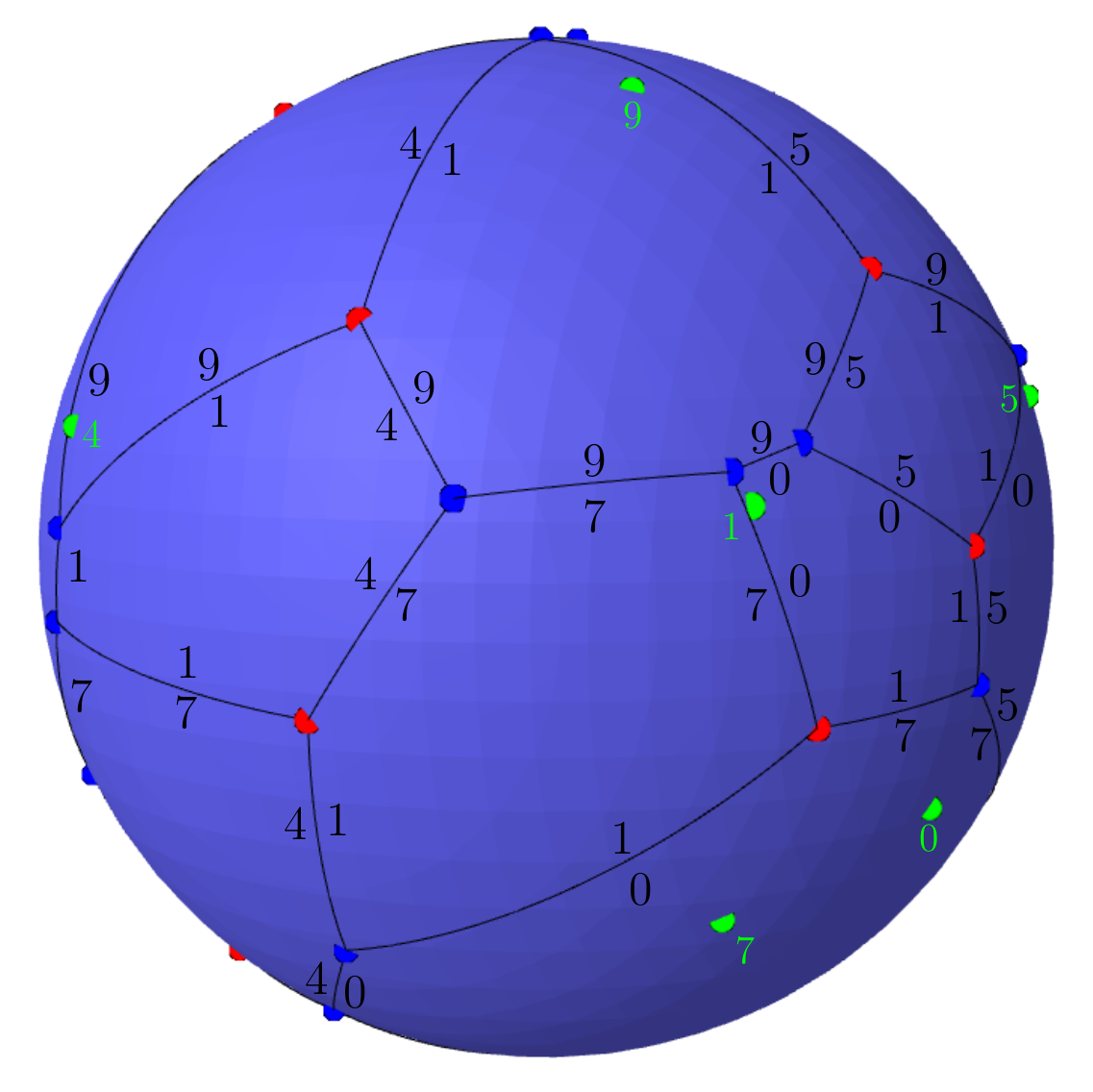}
    		\caption{The edge labeling of $SV_2(U)$ for a set $U$  of ten points $\{0,1,\ldots, 9\}$ in general position (the visible ones are drawn in green color). Vertices of type I are drawn in blue, and vertices of type II in red.}
    		\label{fig:F-bisectors}
    	\end{center}
    \end{figure}
    
    Note that when walking along the boundary of a face, in its interior (exterior), a change in the labels of its edges appears whenever we reach a vertex of type II (type I), see Figure~\ref{fig:F-bisectors}.
    
    \begin{figure}[h!]	
    	\begin{center}
    		\includegraphics[width=0.4\textwidth]{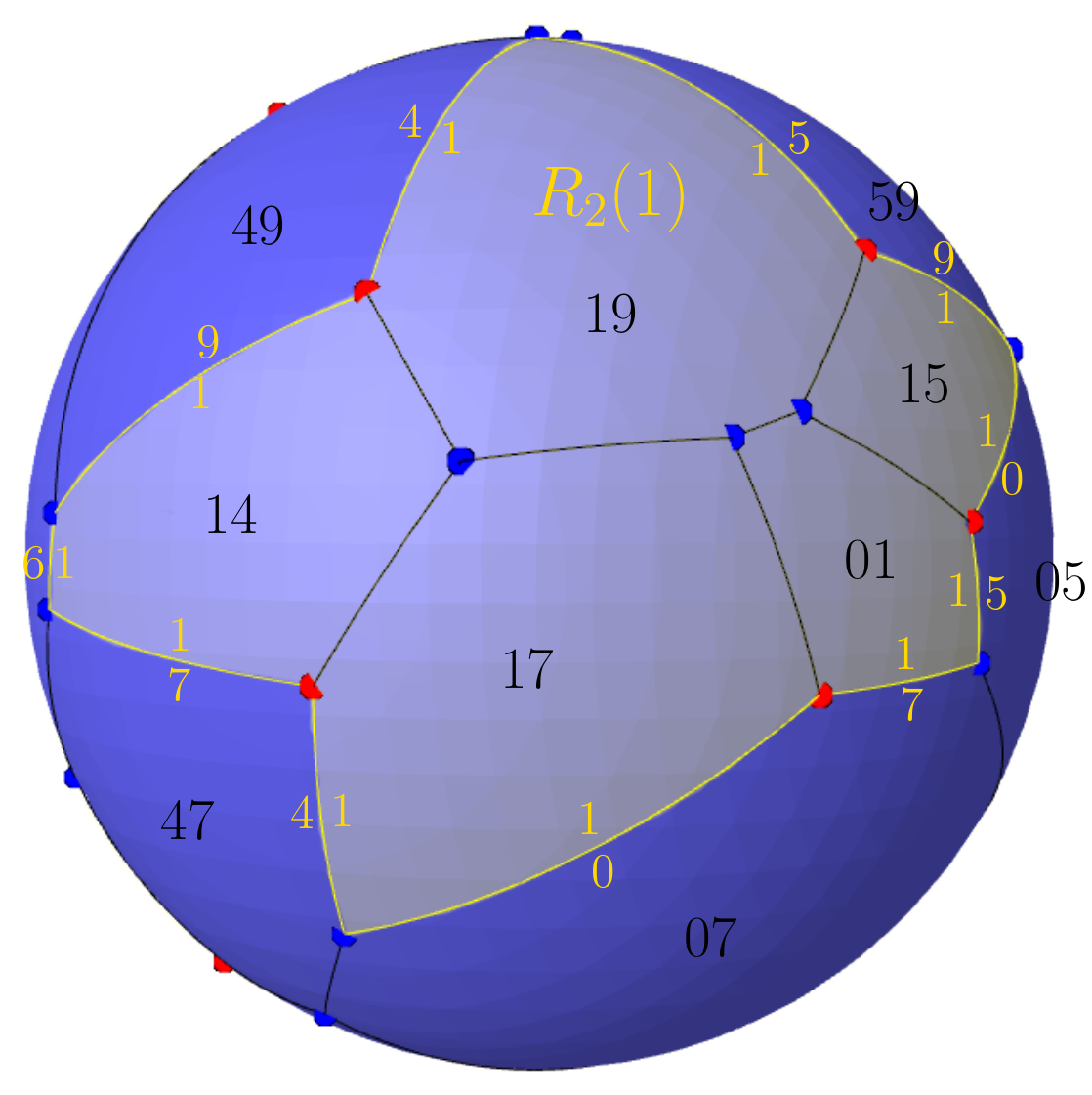}
    		\caption{$SV_2(U)$ for the  point set $U$ of Figure~\ref{fig:F-bisectors}; in each face, its two nearest neighbors are indicated. In yellow, the region $R_2(1)$ formed by all the faces of $SV_2(U)$ that have point $1$ as one of their two nearest neighbors. The boundary $B_2(1)$ of $R_2(1)$ is formed by all the edges which have the label $1$ and this label is always inside $R_2(1)$. The boundary vertices of $R_2(1)$ with an incident edge lying in the interior of $R_2(1)$ are of type II in $SV_2(U)$ and the remaining boundary vertices are of type I in $SV_2(U)$.}
    		\label{fig:R21}
    	\end{center}
    \end{figure}
    
    From this edge labeling, we observe that edges with same label $i$ always form a cycle in $SV_k(U)$; see Figure~\ref{fig:R21}. These edges with the same label $i$ enclose a region $R_k(i)$ that consists of all the points of the sphere that have point $i \in U$ as one of their $k$ nearest neighbors from $U$. 
Such regions $R_k(i)$ are closely related to Brillouin zones~\cite{J84, V00}, also called $k$-th nearest point Voronoi diagram~\cite{OBSC2000}, or $k$-th degree Voronoi diagram~\cite{E87}, defined as $R_k(i)\backslash R_{k-1}(i).$

In~\cite{EI18} it was proved that $R_k(i)$ is star-shaped and in~\cite{MACA21, dH21} it was proved that  the reflex (convex) vertices on the boundary $B_k(i)$ of $R_k(i)$ are vertices of type II (type I). 
We obtain here a new property of the regions $R_k(i)$ of spherical Voronoi diagrams:
As one of the main results of this paper, we generalize to any order the construction of spherical Voronoi diagrams defined by Hyeon-Suk Na, Chung-Nim Lee and Otfried Cheong~\cite{NLC02}, using precisely the regions $R_k(i)$ and the inversion transformation. Inversions for Voronoi diagrams were already applied in the classical work of Brown~\cite{B79, B78}. In~\cite{NLC02}, $SV_1(U)$ is computed from two planar Voronoi diagrams after applying inversions to map $U$ to the plane; two different inversion centers are used. In~\cite{NLC02} it is also shown that $SV_1(U)$ is homeomorphic to the union of a nearest and a farthest Voronoi diagram, when glued together. We generalize this to $SV_k(U)$ being homeomorphic to the union of a planar Voronoi diagram of order $k$, and one planar Voronoi diagram of order $n-k$. Furthermore, these diagrams are linked via $R_k(i)$ in $SV_k(U \cup \{i\})$, with $i$ the center of inversion, where the unbounded edges in the two planar Voronoi diagrams correspond to edges of $SV_k(U)$ intersected by $B_k(i)$. See Figure~\ref{fig:extra}.
		\begin{figure}[h!]
	\centering
	\includegraphics[scale=0.75]{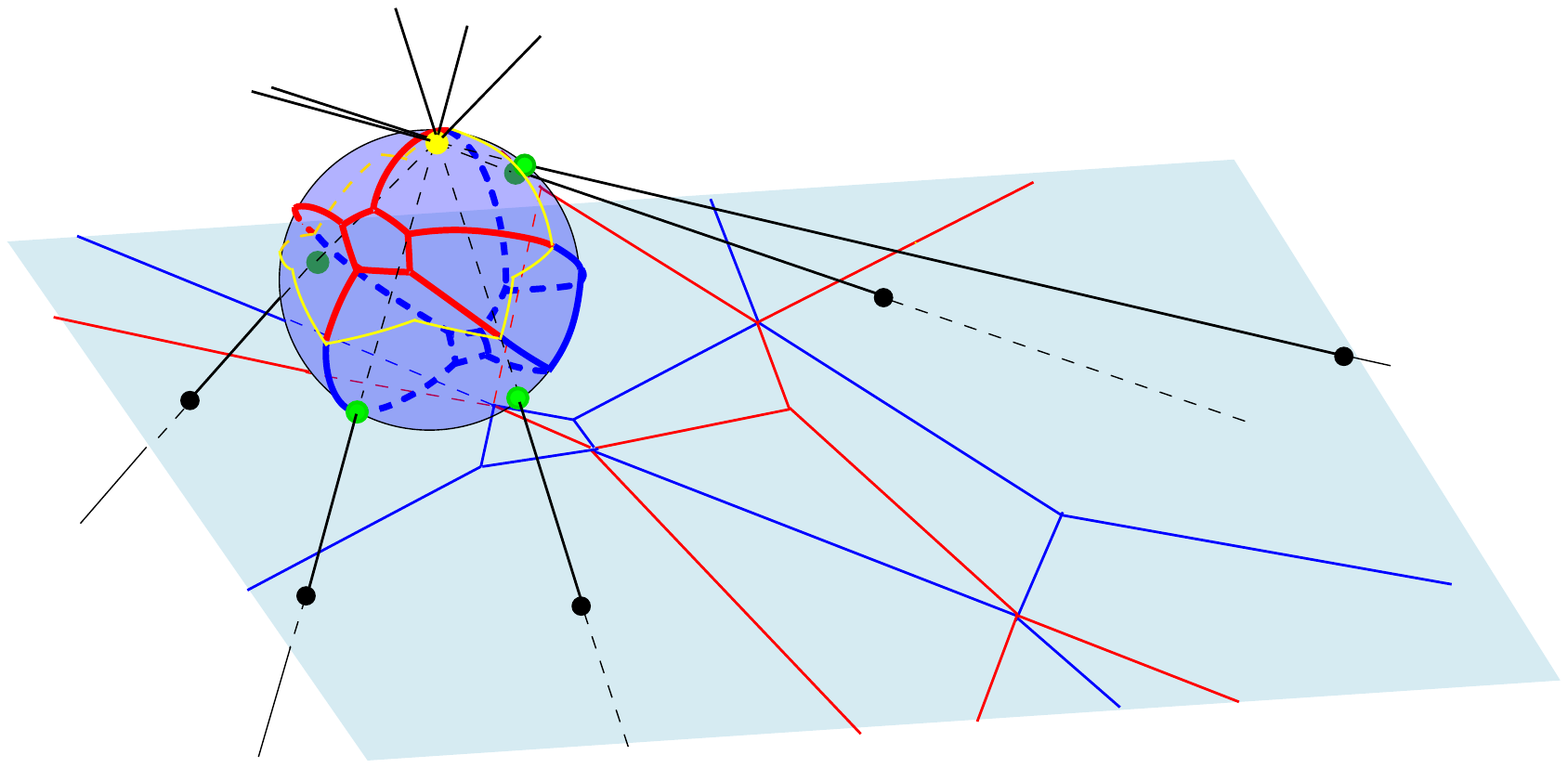}
	\caption{ $SV_k(U)$ (with $k=2$) for a set $U$ of five points (in green) on the sphere, and the inversion transformation and the homeomorphism between $SV_k(U)$ and two Voronoi diagrams in the plane of complementary order ($k=2$ in blue and $n-k=3$ in red). The unbounded edges in the plane correspond to the edges of $SV_2(U)$ intersected by $B_2(i)$ (yellow arcs), where $i$ is the inversion center (yellow point).}
	\label{fig:extra}
\end{figure}
		
    We further derive formulas for the numbers of vertices, edges, and faces of $SV_k(U)$. The proof is based on the construction of $SV_k(U)$. Surprisingly, the obtained formulas seem to be new. 
    We also obtain formulas for the number of vertices of type I and for the number of vertices of type II in $SV_k(U)$. 
	
		As another application of the edge labeling we show that $SV_k(U)$ admits a small {\it{cycle double cover}}.
		 A cycle double cover~\cite{J85} of a graph $G$ is a collection of cycles $\mathcal{C}$ such that every edge of $G$ belongs to precisely two cycles of $\mathcal{C}$. A double cover $\mathcal{C}$ is orientable if an orientation can be assigned to each element of $\mathcal{C}$ such that for every edge $e$ of $G$, the two cycles that cover $e$ are oriented in opposite directions. Much research was done on finding small cycle double cover for several classes of graphs, see for instance~\cite{Bondy90, S92, S93}.
    We show that every higher-order Voronoi diagram on the sphere admits an orientable double cover of its edges, using, precisely, the $n$ cycles $B_k(i)$ for $i=1, \ldots, n$. We refer to~\cite{MACA21} for related results on double covers of the edges of higher order Voronoi diagrams in the plane.
 Several more properties for $SV_k(U)$ are given in the thesis~\cite{dH21} by the second author.
   
\section{Properties of $SV_k(U)$}\label{sec:properties}
      
We present several properties for $SV_k(U)$ that will be used for the main result in Section~\ref{sec:plane-sphere}.
				
    \begin{property}\label{prop:complementaryVoronois}
    		Let $u^*$ be the antipodal point of a point $u$ on a sphere $S$.
    		Then $SV_k(U)=SV_{n-k}(U^*)$, where $U^*=\{ u^* : u\in U\} $.
    \end{property}
    
    The proof of this property is essentially the same as the one for the case $k=1$ given in~\cite{B79, NLC02}.
    
     \begin{proof}
         The spherical distance for points $x,y\in S$ is  $d(x,y)=\pi r-d(x,y^*)$, where $r$ is the radius of the sphere. It follows that the $k$ nearest neighbors of a point $x$ must be the $k$ farthest neighbors of $x^*$. Therefore, $x\in f(P_k)$ if and only if $x\in f(U^*\setminus P_k^*)$ where $P_k^*=\lbrace p^*:p\in P_k\rbrace$, and the property follows.
     \end{proof}
    
    \begin{property}\label{prop:antipodalstype}
        Let $v$ be a vertex of type~I of $SV_k(U)$. Then $v^*$ is a vertex of type~II of $SV_{n-k}(U)$. Similarly, if $v$ is a vertex of type~II of $SV_k(U)$ then $v^*$ is a vertex of type~I of $SV_{n-k}(U)$. See Figure~\ref{fig:CompltypePoints}.
    \end{property}
    \begin{proof}
    If $v$ is a vertex of type~I of $SV_k(U)$, then it is the center of a disk $D$ that passes through three points of $U$ and contains $(k-1)$ points of $U$ in its interior. From this, by the geometry of the sphere $S$, the remaining $(n-k-2)$ points are contained in the complementary disk $S\setminus D$ whose center is $v^*$. Therefore, $v^*$ must be a vertex of type~II of $SV_{n-k}(U)$. The symmetric argument works for $v$ of type~II. 
    \end{proof}
    
    \begin{figure}[h!]	
    		\begin{center}
    			\includegraphics[width=1\textwidth]{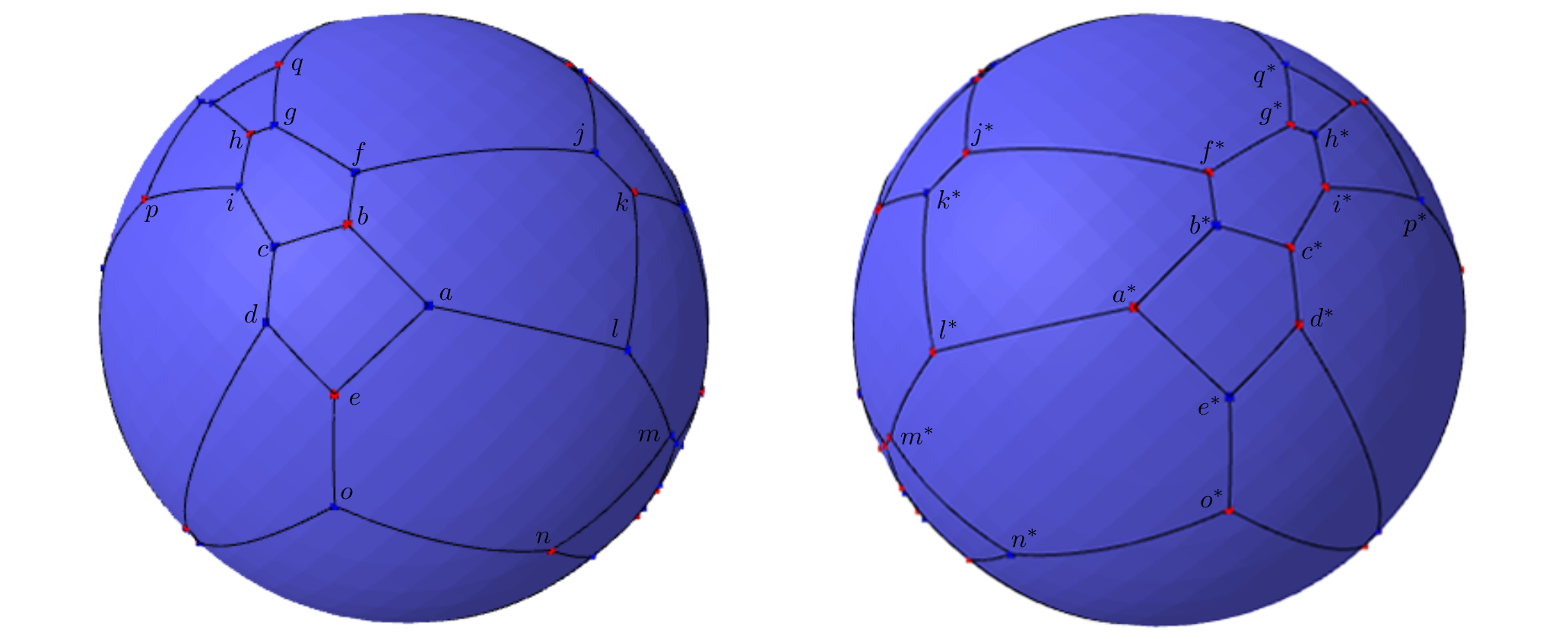}
    			\caption{Two complementary Voronoi diagrams on an sphere $SV_k(U)$ and $SV_{n-k}(U)$, showing the homothetic relation between them and their corresponding antipodal points types. Type~I vertices are blue and type~II vertices are red.}
    			\label{fig:CompltypePoints}
    		\end{center}
    	\end{figure}
    
    \begin{property}\label{prop:complentaryLabelings}
        Let $f(P_k)$ be a face of $SV_k(U)$ and let $f(U\setminus P_k)$ be its corresponding antipodal face in $SV_{n-k}(U)$. $f(P_k)$ and $f(U\setminus P_k)$ use the same labels but in opposite sides, i.e., if $i$ is an interior label of an edge of $f(P_k)$ then it is an exterior label for the corresponding antipodal edge in $SV_{n-k}(U)$. See Figure~\ref{fig:AntipodalLabel}.
    \end{property}
    \begin{proof}
    It follows from Property~\ref{prop:complementaryVoronois} that $f(P_k)$ and $f(U\setminus P_k)$ are antipodal polygons. Then we just need to observe that antipodal polygons are defined by the same bisectors, namely as the intersection of the complementary half-spheres, i.e, their edges are from the same bisectors but the antipodal polygons lie in opposite sides of those bisectors, see Figure~\ref{fig:AntipodalLabel}. Therefore, the statement of Property~\ref{prop:complentaryLabelings} follows from the edge rule. 
    \end{proof}
        
    \begin{theorem}\label{thm:doublecycle}
    $SV_k(U)$ has an orientable double cover consisting of $\vert U\vert =n$ cycles.
    \end{theorem}
    \begin{proof}
		From the vertex rule, we see that every vertex of $SV_k(U)$ is incident to zero or two edges with label $i$. Then the set of edges with label $i$ forms a union of cycles. The point corresponding to label $i$ is contained inside each such cycle. By the definition of the edge rule, there can only be one such cycle. (Also see Property 6.1 in~\cite{dH21} for a more extended argumentation.)
		
	Since each label
    $i$, corresponding to a point $i \in U$, is inside the corresponding region $R_k(i)$, we can orient all the edges of a cycle with label $i$ clockwise around point $i$; note that point $i$ is also contained in $R_k(i).$  This shows that the cycle cover is orientable. Finally, as there is one cycle for each point of $U$, and since each edge has two labels, it follows that $SV_k(U)$ has an orientable double cover of $n$ cycles.
    \end{proof}   
    
    \begin{figure}[h!]	
    \begin{center}
    	\includegraphics[width=0.5\textwidth]{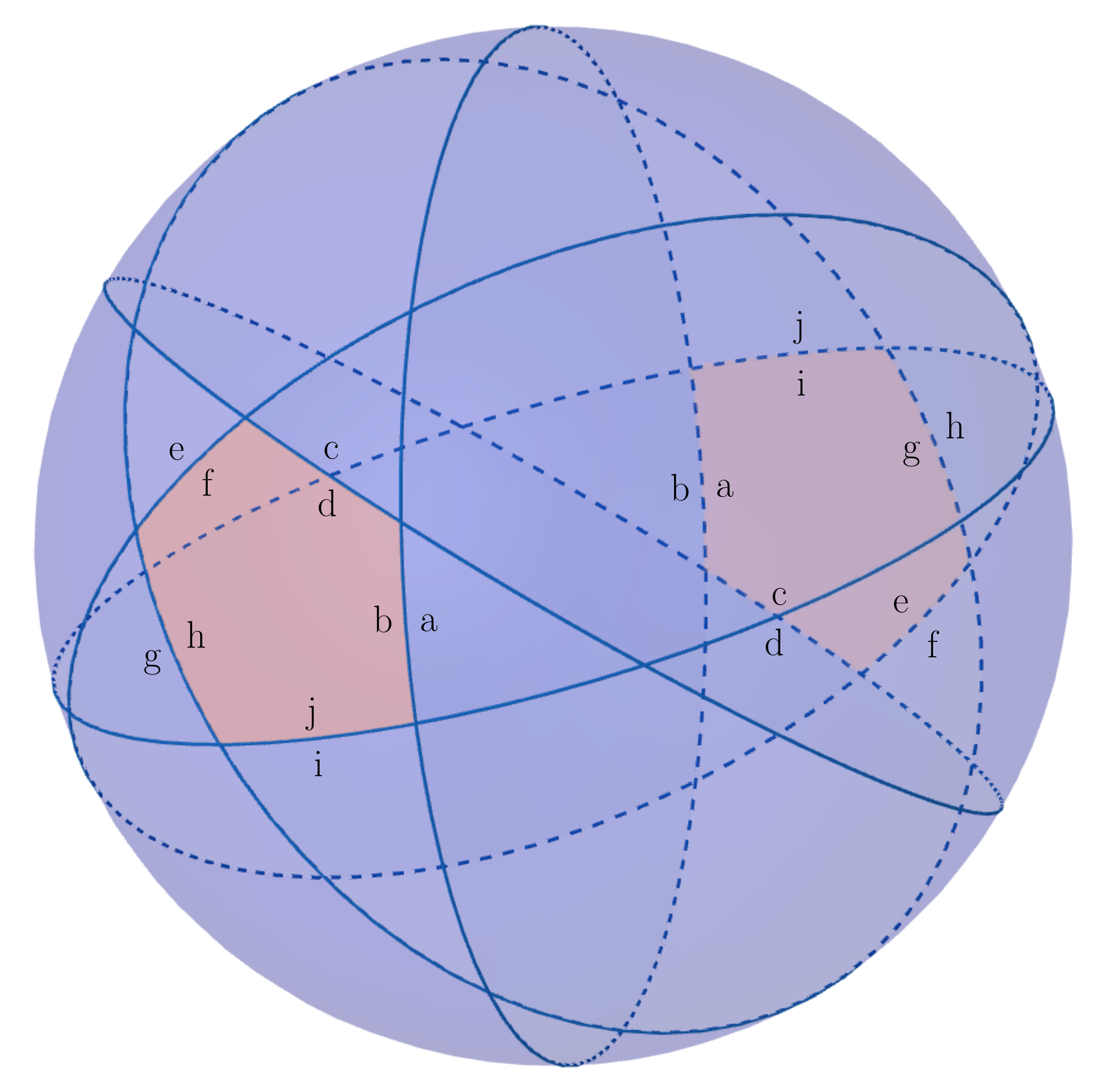}
    	\caption{Two antipodal polygons, one has labels ${b,d,f,h,j}$ in its interior, the other one has these labels in its exterior.}				
    	\label{fig:AntipodalLabel}
    \end{center}
    \end{figure}

\section{Relations between Planar and Spherical Voronoi Diagrams}
    \label{sec:plane-sphere}
    In this section we generalize to Voronoi diagrams of arbitrary order $k$ the construction given in~\cite{NLC02} for the nearest and farthest Voronoi diagrams. We then prove some more properties using this construction. 
    
    First, we need to define the inversion transformation, as it is the basis of the relation between Voronoi diagrams on the sphere and on the plane. 
    
    \begin{definition}
        The inversion transformation is determined by two parameters: The center of inversion $O$ and the radius of inversion $R$. Two points $P$ and $P'$ in $\mathbb{R}^3$ are said to be inverses of each other if:
        \begin{enumerate}
            \item The points $P$ and $P'$ lie in the same half-line with origin in $O$.
            \item  The Euclidean distances $\vert \overline{OP}\vert $ and $\vert \overline{OP'}\vert $ in $\mathbb{R}^3$ satisfy $R^2=\vert \overline{OP}\vert  \vert \overline{OP'}\vert $.
        \end{enumerate}
    \end{definition}
    
    For this transformation it is important to remark the following well-known properties of inversions in $\mathbb{R}^3$ as they are the key of the constructions that we define. For more details about the inversion and its properties see \cite{BAK75}.

    \begin{property}\label{prop:circleNotPass}
        The inverse of any sphere $S$ that does not pass through the center of inversion is a sphere $S'$ that does not pass through the center of inversion. Also, if the center of inversion is not at the interior of $S$, then the interior of $S$ transforms to the interior of $S'$ and the exterior of $S$ transforms into the exterior of $S'$.
    \end{property}
    
    \begin{property}\label{prop:circlePass}
        The inverse of any sphere that passes through the center of inversion is a plane that does not pass through the center of inversion. And the interior of such a sphere transforms to a half-space bounded by that plane and the exterior of the sphere transforms to the opposite half-space.
    \end{property}
    
    Now, we can proceed in a similar way to \cite{B78} to prove the construction for Voronoi diagrams on the sphere of arbitrary order, $SV_k(U)$. From now on, we denote by $S'$ the plane inverse of the sphere $S$, by $U'$ the set of points on the plane $S'$ that are inverses of the points of $U\subset S$, and by $V_k(U')$ the Voronoi diagram of order $k$ in the plane for the set of points $U'$.
    
    \begin{theorem}\label{prop:construction}
        Let $i\notin U$ be a point on the sphere $S$ such that $U\cup\{i\}$ is in general position.
        Let $U'$ be the set of inverse points of $U$ for a chosen inversion radius $r$ and $i$ the center of inversion. Then $SV_k(U)$ is homeomorphic to the union of $V_k(U')$ and $V_{n-k}(U')$, joined by the unbounded edges common to $V_k(U')$ and $V_{n-k}(U')$ (unbounded edges from the same bisector are glued together).
        Moreover, $R_k(i)$ in $SV_k(U\cup\{i\})$ partitions $SV_k(U)$ into two subgraphs that are homeomorphic to $V_k(U')$ and $V_{n-k}(U')$. The vertices of type I (type II) in $V_k(U')$ correspond to the vertices of type I (type II) in $SV_k(U)$ and the vertices of type I (type II) in $V_{n-k}(U')$ correspond to the vertices of type II (type I) in $SV_k(U)$. See Figures~\ref{fig:Vk(U')} and~\ref{fig:Vn-k(U')}.
    \end{theorem}
    \begin{proof}
    
        By Property~\ref{prop:circlePass}, since $i$ is a point of $S$, the inverse of the sphere is a plane $S'$. 
        Let $abc$ be a vertex of $SV_k(U)$. By definition, vertex $abc$ is the center of a circle passing through points $a$, $b$ and $c$ of $U$ that either encloses $k-1$ or $k-2$ points of $U$. This circle is contained in a sphere $S_{abc}$ with the same center and radius. Points in the edges of $SV_k(U)$ are centers of circles that pass through two points of $U$ and enclose $k-1$ of $U$, and, in the same way, all such circles are contained in a sphere with the same center and radius. We denote with $C$ a circle centered in a vertex $abc$ or an edge of $SV_k(U)$, where $C$ passes through three, respectively two, points of $U$. The inversion maps $C$ to a circle $C'$ or a line in the plane. In order to obtain the correspondence between vertices and edges of $SV_k(U)$ and $V_k(U')$, respectively $V_{n-k}(U')$, we determine how many points of $U'$ are enclosed by $C'$.
        We distinguish three cases:
        
        Case 1: $C$ does not enclose or pass through the inversion center.\\
        Since the inversion center $i$ is at the exterior of $S$, $i$ is also at the exterior of the corresponding sphere with the same center and radius as $C$. In order to prove that there is an one-to-one correspondence between vertices and edges of $SV_k(U)$ whose defining circles $C$ satisfy the condition of Case 1, and the vertices and bounded edges of $V_k(U')$, we use Property~\ref{prop:circleNotPass}.
				
        First, if the circle $C$ has center $abc$ and passes through $a$, $b$ and $c$, so also does the sphere $S_{abc}$. 
        By Property~\ref{prop:circleNotPass}, the inverse of $S_{abc}$ is a sphere whose intersection with $S'$ is a circle $C'$ that passes through $a'$, $b'$ and $c'$ (the inverse points of $a$, $b$ and $c$) and encloses the inverse of the points in the interior of $S_{abc}$. Thus, $C'$ is a circle with center $a'b'c'$ that passes through the inverse points of $a$, $b$ and $c$ and encloses the same number of points ($k-1$ or $k-2$) of $U'$ as $C$. Therefore, $a'b'c'$ is a vertex of $V_k(U')$ and it is of the same type as $abc$. From that, we have that there is an one-to-one correspondence between vertices of $SV_k(U)$ that satisfy the condition of Case 1, and the vertices of $V_k(U')$. 
        Second, since the points in the edges of $SV_k(U)$ are centers of circles that pass through two points of $U$ and enclose $k-1$ points of $U$, if they do not contain or pass through the inversion center we can proceed in the same way to show that the inverses of such edges of $SV_k(U)$ are edges of $V_k(U')$. See Figure~\ref{fig:Vk(U')}.
        
        Case 2: $C$ encloses the inversion center.\\
        If $C$ encloses the inversion center then the corresponding sphere with the same center and radius as $C$, does so too. In order to prove that there is an one-to-one correspondence between vertices and edges of $SV_k(U)$ whose defining circles $C$ satisfy the condition of Case 2, and the vertices and bounded edges of $V_{n-k}(U')$, we use Properties~\ref{prop:complementaryVoronois} and \ref{prop:complentaryLabelings} to reduce Case 2 to Case 1.\\
        First, if the circle $C$ has center $abc$ and passes through points $a$, $b$, and $c$ of $U$, so also does the sphere $S_{abc}$.
        By Property~\ref{prop:circleNotPass}, the inverse of $S_{abc}$ is a sphere whose intersection with $S'$ is a circle that passes through $a'$, $b'$ and $c'$. But, as it contains $i$, the number of points of $U'$ that are enclosed is not the same as the number of points of $U$ enclosed by $S_{abc}$. Note that $abc$ is a point in the interior of $R_k(i)$ because if $C$ encloses $i$, this means that $i$ is one of the $k$ nearest neighbors of $abc$ in $SV_k(U\cup\{i\})$. By Properties \ref{prop:complementaryVoronois} and \ref{prop:complentaryLabelings}, the boundary of the antipodal polygon of $R_k(i)$ has label $i$ in the exterior, i.e., the interior of the antipodal polygon of $R_k(i)$ is the exterior of $R_{n-k}(i)$. Then, the induced graph by $SV_k(U)$ in the interior of $R_k(i)$ in $SV_k(U\cup\{i\})$ is homeomorphic to the induced graph by $SV_{n-k}(U)$ in the exterior of $R_{n-k}(i)$. See Figure~\ref{fig:R_k(i)-R_n-k(i)}. From that, for $SV_{n-k}(U)$ we are in Case 1, and $V_{n-k}(U')$ is homeomorphic to the induced graph by $SV_{n-k}(U)$ in the exterior of $R_{n-k}(i)$ and, therefore, is homeomorphic to the induced graph by $SV_k(U)$ in the interior of $R_k(i)$. By Property~\ref{prop:antipodalstype}, the corresponding vertex to a type I (type II) vertex in $V_{n-k}(U')$ is of type II (type I) in $SV_k(U)$. See Figure~\ref{fig:Vn-k(U')}.\\
        Second, the same argument applies for $C$ centered at a point in an edge of $SV_k(U)$.
        
        Case 3: $C$ passes through the inversion’s center.\\
        If $C$ passes through the inversion center $i$, then the corresponding sphere with the same center and radius as $C$ does so too. In order to prove that there is an one-to-one correspondence between edges of $SV_k(U)$ whose defining circles $C$ satisfy the condition of Case 3 (such edges contain the center of $C$) and the unbounded edges of $V_k(U')$ and $V_{n-k}(U')$, we use Property~\ref{prop:circlePass}.
				
        Note that $abc$ cannot be the center of a circle passing through $i$, $a$, $b$ and $c$, because $U\cup\{i\}$ is in general position. Therefore, we can assume that the center of the circle that passes through $i$ must be of the form $abi$ where $a$ and $b$ are points of $U$, i.e., $abi$ is a vertex of a Voronoi diagram of some order for $U\cup\{i\}$. And $abi$ is also a point on an edge of $SV_k(U)$. The points on the edges of a Voronoi diagram of order $k$ are centers of circles that pass through two points of $U$ and enclose $k-1$ points of $U$. Then $abi$ must be a vertex of type I in $SV_k(U\cup\{i\})$. 
        By Property~\ref{prop:circlePass}, the inverse of the sphere $S_{abi}$ is a plane that intersects with $S'$ in a line that passes through $a'$ and $b'$.
        Finally, since lines can be regarded as circles with center at infinity, such point $abi$ corresponds precisely to the point at infinity of the perpendicular bisector between $a'$ and $b'$. 
        Note that, since these points on the edges of $SV_k(U)$ correspond to points at infinity, those edges split into the unbounded edges of $V_k(U')$ and $V_{n-k}(U')$. See Figures~\ref{fig:Vk(U')} and~\ref{fig:Vn-k(U')}.
    \end{proof}
    
    \begin{figure}[h!]
    	\centering
    		\includegraphics[scale=0.5]{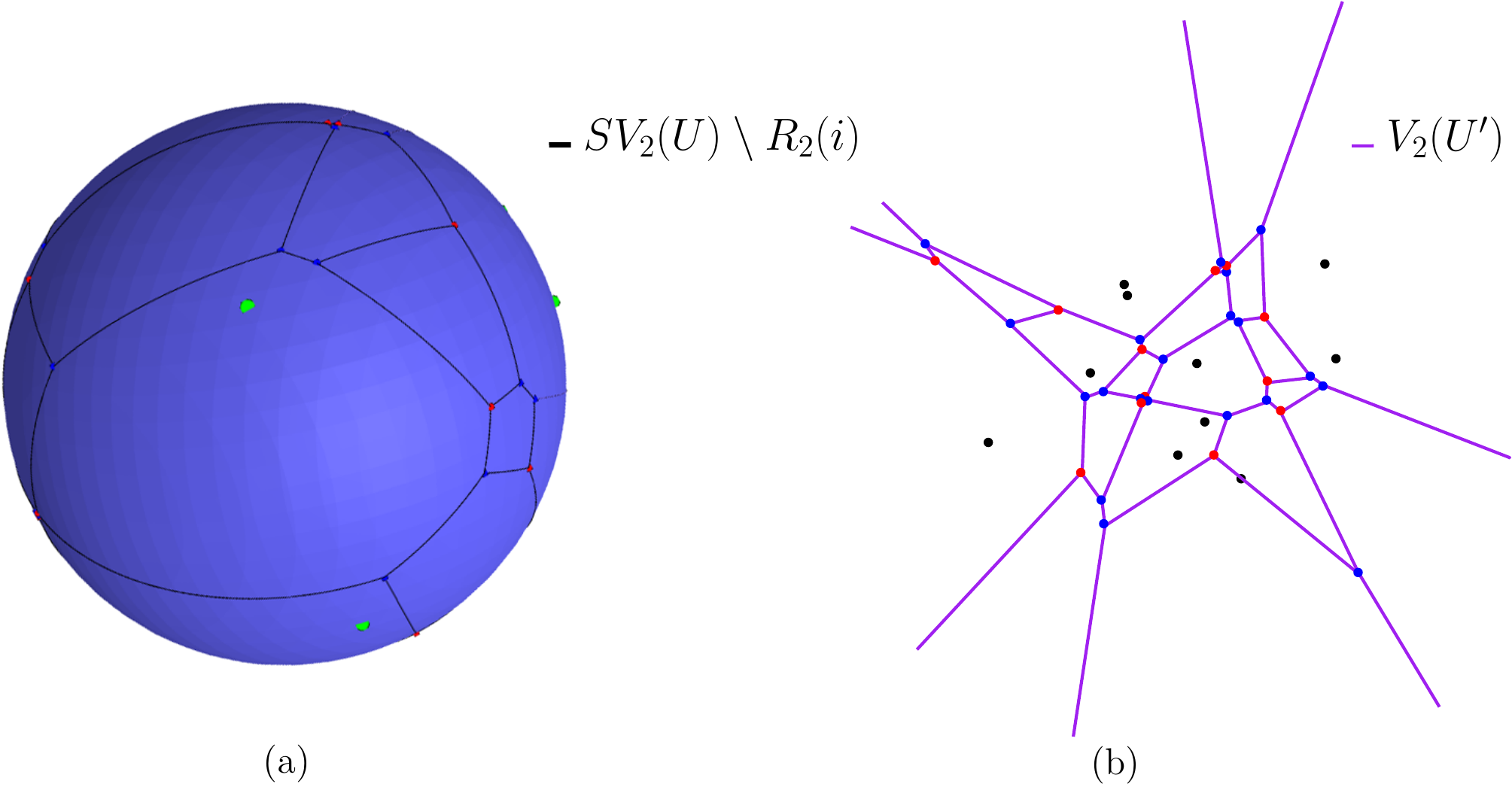}
    	\caption{For a set $U$ of ten points on the sphere (the visible ones are drawn in green color): The picture shows the homeomorphism between: (a) The induced graph by $SV_2(U)$ at the exterior of $R_2(i)$ in $SV_2(U\cup\{i\})$. (b) The planar Voronoi diagram of order ${2}$ for the points of $U'$ (black color).}
    	\label{fig:Vk(U')}
    \end{figure}
    
    \begin{figure}[h!]
    	\centering
    		\includegraphics[scale=0.5]{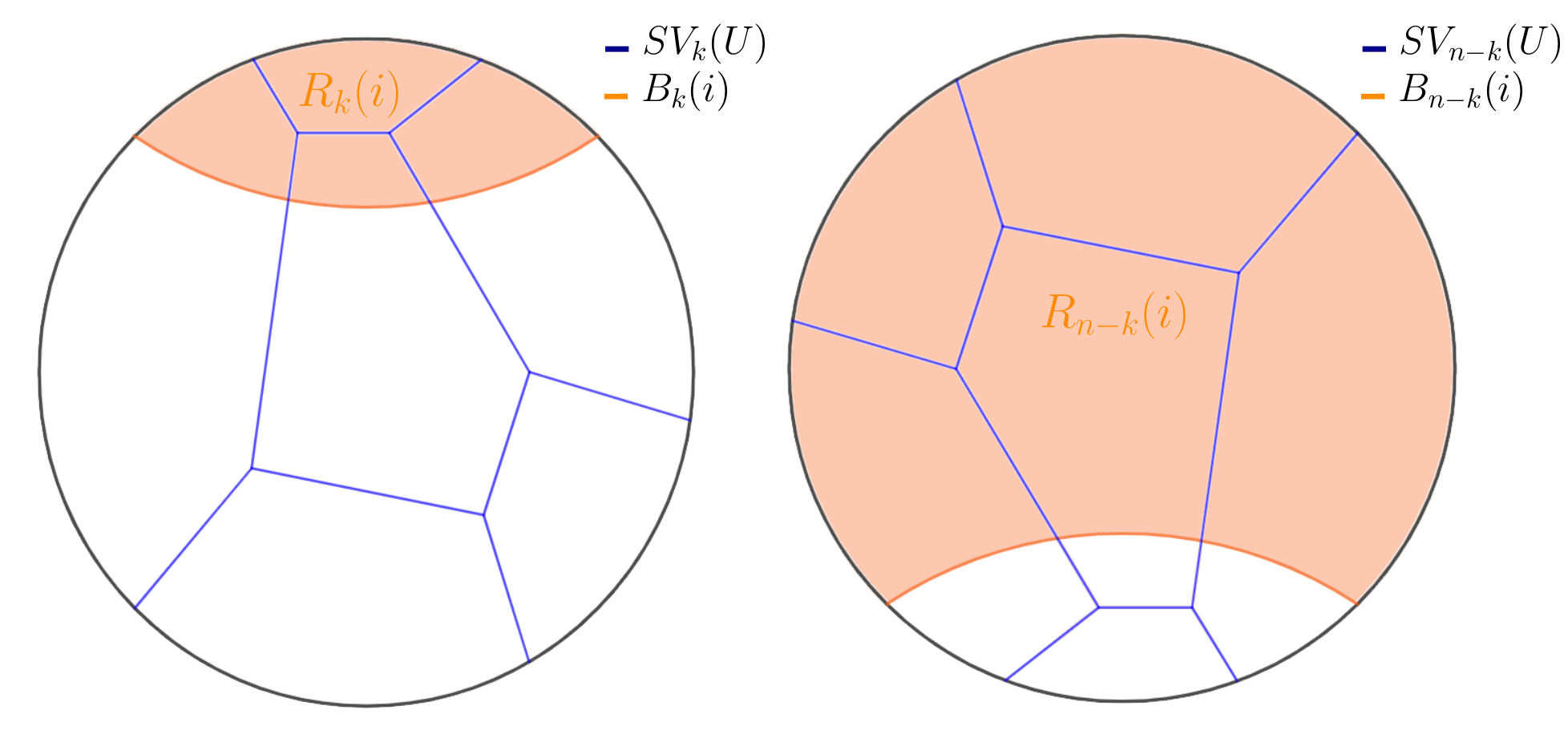}
    	\caption{For a set $U$ of points on the sphere. The picture shows the homeomorphism between: (a) The induced graph by $SV_k(U)$ at the interior of $R_k(i)$ in $SV_k(U\cup\{i\})$. (b) The induced graph by $SV_{n-k}(U)$ at the exterior of $R_{n-k}(i)$ in $SV_{n-k}(U\cup\{i\})$.}
    	\label{fig:R_k(i)-R_n-k(i)}
    \end{figure}
    
    \begin{figure}[h!]
    	\centering
    		\includegraphics[scale=0.51]{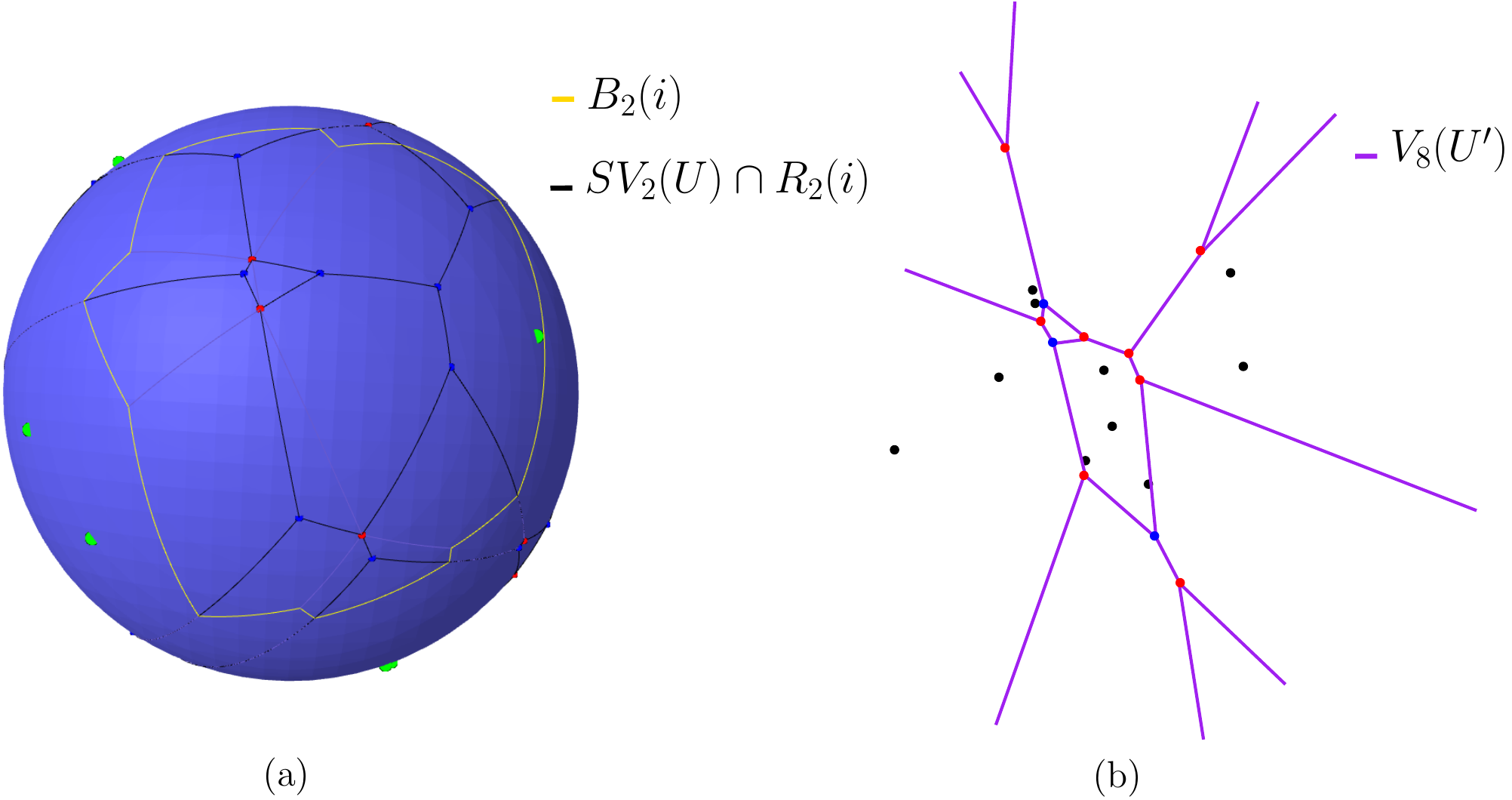}
    	\caption{For a set $U$ of ten points on the sphere (the visible ones are drawn in green color): The picture shows the homeomorphism between: (a) The induced graph by $SV_2(U)$ at the interior of $R_2(i)$ in $SV_2(U\cup\{i\})$. (b) The planar Voronoi diagram of order ${8}$ for the points of $U'$ (black color).}
    	\label{fig:Vn-k(U')}
    \end{figure}
    
    Theorem~\ref{prop:construction} tells us how to construct $SV_k(U)$: we just have to invert the points of $U$, compute planar Voronoi diagrams $V_k(U')$ and $V_{n-k}(U')$, and map them to the sphere as follows: each vertex $a'b'c'$ of either $V_k(U')$ or $V_{n-k}(U')$ corresponds to a vertex $abc$ of $SV_k(U)$ ($abc$ is center of the circle that passes through $a$, $b$ and $c$ on the sphere); vertices in $SV_k(U)$ are adjacent whenever the corresponing vertices in $V_k(U')$ or in $V_{n-k}(U')$ are adjacent.
    Finally, the vertices of $SV_k(U)$ corresponding to vertices incident to an unbounded edge from the same bisector in  $V_k(U')$ and $V_{n-k}(U')$ get connected.
    
 Let us shortly comment on the computational complexity of constructing higher order Voronoi diagrams on the sphere. The inversion is a linear time transformation and, once we have the planar Voronoi diagrams, mapping them to the sphere also only requires linear time. Therefore, the computational time for constructing the spherical Voronoi diagrams is bounded by the computational time for the planar ones. 
    
    Now, from these constructions, it is easy to see that properties proved for the plane \cite{MACA21} must be true for the sphere. We can prove easily some properties on the sphere using results from the plane, but also we can prove properties in the plane using the sphere. Next, we show that the number of vertices of type I (type II) in $SV_k(U)$ only depends on the number $n$ of points of $U$, but not on their positions on the sphere.
    
    \begin{theorem}\label{prop:Ntype}
        For a set $U$ of $n$ points on the sphere, the number of vertices of type I in $SV_k(U)$ is $2k(n-k-1)$ and the number of vertices of type II is $2(k-1)(n-k)$.
    \end{theorem}
    \begin{proof}
        By Theorem~\ref{prop:construction}, we can define an inversion transformation such that there is a one-to-one correspondence between the vertices of $SV_k(U)$ and the vertices of $V_k(U')$ and $V_{n-k}(U')$. Vertices of type I of $SV_k(U')$ and vertices of type II of $V_{n-k}(U')$ correspond to the vertices of type I in $SV_k(U)$. Then, the number of vertices of type I in $SV_k(U)$ is the sum of type I vertices of $V_k(U')$ and type II vertices of $V_{n-k}(U')$ which correspond to the circles enclosing $k-1$ points of $U'$ and circles enclosing $n-k-2$ points of $U'$, respectively. We denote the number of such circles with $c_{k-1}$ and $c_{n-k-2}$. By Theorem 5.3 of \cite{L03}, also see~\cite{ardila2004number, clarkson1989applications,L82}, we have
        \begin{equation}\label{eq1}
             c_{k-1}+c_{n-k-2}=2(k-1+1)(n-2-k+1)=2k(n-k-1).
        \end{equation}
        
        Then, the number of vertices of type I in $SV_k(U)$ is $2k(n-k-1)$. Similarly, we can compute the number of vertices of type II as the sum of vertices of type II in $V_k(U')$ and type I in $V_{n-k}(U')$, i.e., the number of the circles enclosing $k-2$ points of $U'$, $c_{k-2}$, and enclosing $n-k-1$ points of $U'$, $c_{n-k-1}$. Again, using Theorem 5.3 of \cite{L03}, we have
         \begin{equation}\label{eq1}
             c_{k-2}+c_{n-k-1}=2(k-2+1)(n-2-k+2)=2(k-1)(n-k).
        \end{equation}
        Then, the number of vertices of type II in $SV_k(U)$ is $2(k-1)(n-k)$.
    \end{proof}
    
    \begin{theorem}\label{prop:NedgesNfaces}
        For a set $U$ of $n$ points on the sphere, the order $k$ Voronoi diagram $SV_k(U)$ has $4kn-4k^2-2n$ vertices, $6kn-6k^2-3n$ edges and $2kn-2k^2-n+2$ faces.
    \end{theorem}
    \begin{proof}
        Vertices of spherical Voronoi diagrams are either of type I or type II, so the total number of vertices is the sum of vertices of the two types. Then, by Theorem~\ref{prop:Ntype}, the number of vertices $\vert  V\vert $ is 
        \begin{equation}\label{eq:vertices}
            \vert V\vert =2k(n-k-1)+2(k-1)(n-k)=4kn-4k^2-2n.
        \end{equation}
        
				We count the number of edges $\vert E \vert$ of $SV_k(U)$ by double counting. Note that each vertex has degree three in $SV_k(U).$ Since each edge is incident to two vertices, we get 
%
        \begin{equation}\label{eq:edges}
            \vert E\vert =\frac{3}{2}\left(-4k^2+4kn-2n\right)=6kn-6k^2-3n.
        \end{equation}
        Finally, as $SV_k(U)$ is a planar graph, we can apply Euler's formula to count the number of faces $\vert F\vert $, and we have 
        \begin{equation}\label{eq:faces}
            \vert F\vert =2-(-4k^2+4kn-2n)+(-6k^2+6kn-3n)=2kn-2k^2-n+2.
        \end{equation}
    \end{proof}

    \subparagraph*{Acknowledgments.} Research of C. Huemer and of M. Claverol was supported by project PID2019-104129GB-I00/ MCIN/ AEI/ 10.13039/501100011033.

    {\begin{minipage}[l]{0.3\textwidth} \includegraphics[trim=10cm 6cm 10cm 5cm,clip,scale=0.15]{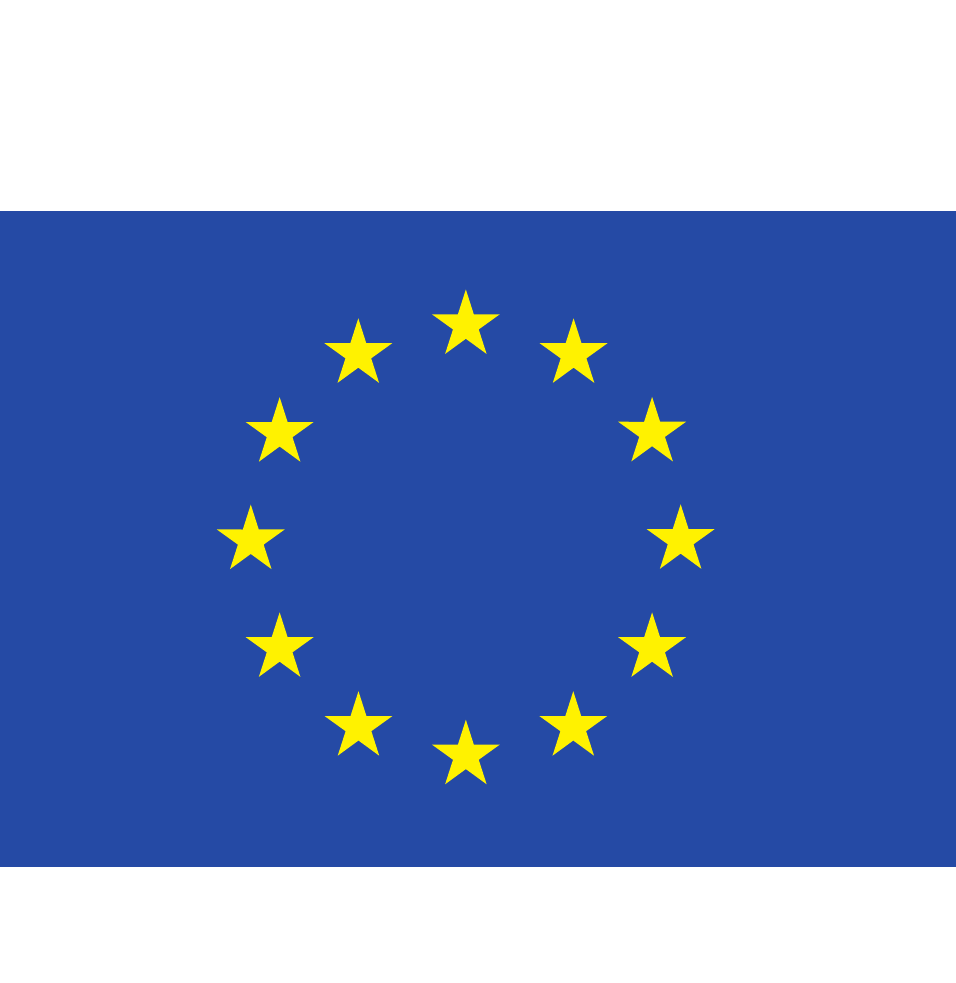} \end{minipage}
     \hspace{-2.2cm} 
     \begin{minipage}[l][0.5cm]{0.8\textwidth}
     	  This project has received funding from the European Union's Horizon 2020 research and innovation programme under the Marie Sk\l{}odowska-Curie grant agreement No 734922.
     \end{minipage}}

\bibliographystyle{plain}
\bibliography{SphericalVoronoiDiagrams}  






\end{document}